\documentclass[runningheads]{llncs}
\usepackage{graphicx,tikz, subcaption} 
\usepackage{amssymb,amsfonts,mathtools}
\usepackage{hyperref, cleveref}
\usepackage[shortlabels]{enumitem}
\usepackage{algorithm}
\usepackage[noend]{algpseudocode}
\usepackage{etoolbox}
\MakeRobust{\Call}
\usepackage{float}
\usepackage{fixltx2e}
\usetikzlibrary{calc,decorations.pathreplacing,arrows}
\definecolor{eqA}{RGB}{30,90,200}
\definecolor{eqB}{RGB}{200,50,50}
\definecolor{e1col}{RGB}{230,140,0}
\definecolor{e2col}{RGB}{40,150,70}
\definecolor{e3col}{RGB}{140,70,180}
\definecolor{e4col}{RGB}{20,150,170}

\newcommand{\Bsq}{b_\Box}

\title{Solving Square-Submatrix Equation Systems}

\author{
Lorenzo Carfagna\orcidID{0009-0005-9591-057X}, \and
Giovanni~Manzini\orcidID{0000-0002-5047-0196}}
\institute{Department of Computer Science, University of Pisa, Italy}

\authorrunning{L. Carfagna, G. Manzini}

\begin{document}

\maketitle

\begin{abstract}
We consider systems of submatrix equations, that is, sets of equality constraints over square submatrices of the input. 
By generalising the recursive algorithm of Gawrychowski et al. [TCS'20] to two dimensions, we obtain a linear-time procedure that finds a solution for any such input system. 
As an immediate by-product, this yields an optimal-time algorithm for decompressing any two-dimensional macro scheme based on copy operations of sub-squares. 

\end{abstract}

\section{Introduction}\label{sec:intro}
A substring equation system (SES) $E$, expresses a set of equality constraints over substrings of a string. 
In~\cite{GAWRYCHOWSKI2020}, Gawrychowski et al. show that a generic solution for any $E$ can be computed in linear $O(\vert E\vert+n)$ time where $\vert E\vert $ the number of input equations and $n$ is the length of the solution.
Such resolution algorithm can be applied to solve in linear time a series of well-known string-reconstruction problems such as recovering a string from its prefix, border and cover arrays~\cite{GAWRYCHOWSKI2020}.
Recently, Shibata and Bannai~\cite{Shibata_Bannai_2026} analysed SESs from the perspective of compressed representations of strings. In this setting we are given a string $w$ and we want to find the size of the minimal SES representing uniquely $w$ where in addition to substrings equalities we can use explicit symbol assignments. 
Shibata and Bannai showed that the size of a minimal SES describing a string is a meaningful repetitiveness measure and they proved that every valid string bidirectional macro scheme~\cite{STORER1982} can be transformed into an equivalent SES of the same size. 
In the same paper, the authors also showed that every SES $E$ induces a string attractor~\cite{KP2018} of size $4\vert E\vert$ and that for the Thue–Morse words the size of a minimal string attractor is asymptotically smaller than the size of a minimal SES. 
Another remarkable result about SESs is related to the repetitiveness measure 
$\chi(w)$ defined as the cardinality of the smallest suffixient set for $w$~\cite{DEPUYDT2024}. Shibata and Bannai showed that 
for every string $w$ we can build a SES of size $O(\chi(w))$ representing~$w$. 
This yields an $O(\chi(w))$ representation of $w$ thus establishing the {\em reachability} of measure~$\chi$.

The above results suggest that SESs constitute an interesting compression tool and deserve further investigations.
In this paper we study square-submatrix equation systems (SSES), that is two-dimensional SESs made of equations over square submatrices of the input. 
We show that the recursive algorithm in~\cite{GAWRYCHOWSKI2020} can be generalised to find a generic solution of any SSES $E$ in $O(\vert E\vert +  mn)$ time where $mn$ is the matrix size.  
Next, following~\cite{Shibata_Bannai_2026} we consider SSESs as compression tools with the inclusion of explicit symbol assignments. We compare SSESs with the two-dimensional macro schemes as defined in~\cite{CarfagnaManzini2024} where the input matrix is parsed into, possibly overlapping square phrases. 
We show that every 2D macro scheme can be naturally transformed into an equivalent SSES of the same size in linear time. This fact, together with our resolution algorithm, shows that every valid two-dimensional macro scheme of size $b$ for an $m\times n$ matrix can be decompressed in optimal $O(b+mn)$ time.






\section{Notation}\label{sec:notation}

We denote with $\Sigma^{m\times n}$ the set of all $m\times n$ matrices from a finite alphabet~$\Sigma$. 
For every integer $n\geq 1$, $[n]$ is the set $\{0,\ldots, n-1\}$. 
We assume matrices are $0$-indexed. Given $A\in\Sigma^{m\times n}$ we write $A[i..i+\ell-1][j..j+\ell'-1]$ to refer the $\ell\times \ell'$ submatrix of $A$ with top-left cell at $A[i][j]$ and bottom-right cell at $A[i+\ell-1][j+\ell'-1]$ for any $1\leq \ell \leq m-i$ and $1\leq \ell' \leq n-j$.
Given any submatrix, we number its vertices $0, 1, 2, 3$ anticlockwise starting from top-left one.

A \emph{square bidirectional macro scheme}~\cite{CarfagnaManzini2024} for an $m \times n$ matrix~$A$ consists of a partition of~$A$ into (possibly overlapping) {square} submatrices, called {\em phrases}, such that each submatrix is either an explicit symbol (i.e. it is has size $1\times 1$) or it is a copy of another submatrix (called source). Because of possible phrase overlaps, the same position $(i,j)$ can be contained in different phrases and its value $A[i][j]$ can be copied from different sources. Formally, we define a {\em multi-valued} function $f\colon [m] \times [n]  \to 2^{[m] \times [n]} \cup \{\bot\}$, where $2^{[m] \times [n]}$ is the power set of $[m] \times [n]$. The value $f(i,j)$ is a set whose size is equal to the number of phrases containing $(i,j)$. For each phrase $A[u..u+\ell-1][v..v+\ell-1]$ containing $(i,j)$: if $\ell=1$ we add 
$\bot$ to $f(i,j)$, otherwise if the phrase is copied from $A[u+d_u..u+d_u+\ell-1][v+d_v..v+d_v+\ell-1]$ we add $(i+d_u,j+d_v)$ to $f(i,j)$. 
The scheme is valid if we can eventually retrieve the value of any entry in $A$, that is, for every $(i,j) \in [m]\times [n]$ there must exists a sequence of pairs $\{ (i_h, j_h)\}_{h=1..k}$ such that
\begin{equation}\label{eq:valid}
(i_1,j_1) = (i,j),\qquad
(i_{h+1},j_{h+1}) \in f(i_{h},j_h),\qquad
f(i_{k},j_k) = \bot.    
\end{equation}
We denote by $\Bsq(A)$ the size of the smallest valid bidirectional scheme for $A$.



\begin{definition} [Square Submatrix Equation System]
A Square Submatrix Equation System (SSES) $E$ for a $m\times n$ matrix, is a set of submatrix equality constraints (equations) of the form $A[i..i+\ell-1][j..j+\ell-1]=A[i'..i'+\ell-1][j'..j'+\ell-1]$ for some $(i,j),(i',j') \in [m]\times [n]$ and $1\leq \ell \leq \min\{m-i,n-j\}$ referred in the following as $(i,j,i',j',\ell)$.
A solution for $E$ is any matrix $A$ satisfying all the constraints in $E$ simultaneously.
\end{definition}

Note that every $m\times n$ matrix containing only one symbol would be a solution. We are interested in a solution  $\Phi(E)$ containing a maximal number of different characters since any other solution can be obtained from $\Phi(E)$ by remapping alphabet symbols.
A solution $\Phi(E)$ can be obtained by means of the position graph $G_{pos}(E)$ introduced in~\cite{GAWRYCHOWSKI2020} for strings.  $G_{pos}(E)$ is an undirected graph whose nodes are the positions $[m]\times[n]$ and whose arcs represent the position-wise equalities described by the equations in $E$.

\begin{example}\label{ex:gpos}
Let $m=n=4$ and let $E$ consists of the equations
$e_A=(0,0,\,0,2,\,2)$ and $e_B=(0,0,\,2,0,\,2)$,
i.e. 
$$
A[0..1][0..1]=A[0..1][2..3] \quad\mbox{and}\quad A[0..1][0..1]=A[2..3][0..1].
$$
Each equation induces $4$ position equalities, so
$G_{{pos}}$ has $8$ edges 
as shown~in~Fig.~\ref{fig:gpos}.
\end{example}

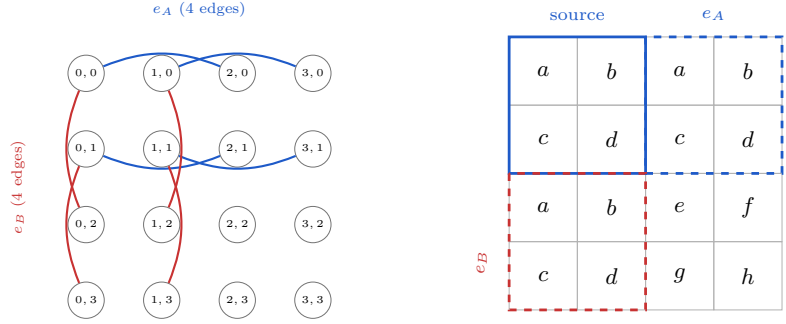
\begin{figure}[t]
\centering
\scalebox{0.80}{
\begin{tikzpicture}[x=1.25cm,y=-1.25cm,baseline=(current bounding box.center)]
  \foreach \r/\c/\s in {
     0/0/a,0/1/b,0/2/a,0/3/b, 1/0/c,1/1/d,1/2/c,1/3/d,
     2/0/a,2/1/b,2/2/e,2/3/f, 3/0/c,3/1/d,3/2/g,3/3/h}{
     \node[circle,draw=black!55,minimum size=0.6cm,inner sep=0pt,fill=white]
        (n\r\c) at (\c,\r) {\tiny $\c,\r$};}
  \begin{scope}[eqA,line width=1pt]
     \draw (n00) to[bend left=22] (n02);  \draw (n01) to[bend left=22] (n03);
     \draw (n10) to[bend right=22] (n12); \draw (n11) to[bend right=22] (n13);
  \end{scope}
  \begin{scope}[eqB,line width=1pt]
     \draw (n00) to[bend right=22] (n20); \draw (n10) to[bend right=22] (n30);
     \draw (n01) to[bend left=22] (n21);  \draw (n11) to[bend left=22] (n31);
  \end{scope}
  \node[eqA] at (1.5,-0.85) {\scriptsize $e_A$ (4 edges)};
  \node[eqB,rotate=90] at (-0.90,1.5) {\scriptsize $e_B$ (4 edges)};
\end{tikzpicture}}
\hspace{1.5cm}
\scalebox{0.95}{
\begin{tikzpicture}[x=0.95cm,y=-0.95cm,baseline=(current bounding box.center)]
  \foreach \r/\c/\s in {
     0/0/a,0/1/b,0/2/a,0/3/b, 1/0/c,1/1/d,1/2/c,1/3/d,
     2/0/a,2/1/b,2/2/e,2/3/f, 3/0/c,3/1/d,3/2/g,3/3/h}{
     \node[draw=black!30,minimum size=0.95cm,inner sep=0pt] at (\c,\r) {$\s$};}
  \draw[eqA,line width=1.1pt] (-0.5,-0.5) rectangle (1.5,1.5);
  \draw[eqA,line width=1.1pt,dashed] (1.5,-0.5) rectangle (3.5,1.5);
  \draw[eqB,line width=1.1pt,dashed] (-0.5,1.5) rectangle (1.5,3.5);
  \node[eqA,anchor=south] at (0.5,-0.62) {\scriptsize source};
  \node[eqA,anchor=south] at (2.5,-0.62) {\scriptsize $e_A$};
  \node[eqB,anchor=east,rotate=90] at (-0.9,2.5) {\scriptsize $e_B$};
\end{tikzpicture}}
\caption{The position graph $G_{pos}(E)$ for Example~\ref{ex:gpos} (left) and the corresponding generic solution $\Phi(E)$ (right).}
\label{fig:gpos}
\end{figure}

We can obtain a solution $\Phi(E)$ by computing the connected components of $G_{pos}(E)$ and by assigning a distinct symbol to each connected component. 
In the example shown in Figure~\ref{fig:gpos}, $G_{{pos}}(E)$ has $8$ connected components so the generic solution is the one in Figure~\ref{fig:gpos} (right).
This procedure is correct since any two distinct positions of $\Phi(E)$ contain the same symbol, if and only if they belong to the same connected component of the position graph. 
Since an equation $(i,j,i',j',\ell)$ generates $\ell^2$ edges, the size of the position graph could be $\Theta(mn\vert E\vert)=\omega(mn)$ so the above approach could be inefficient with respect to the optimal solution of cost  $O(|E|+mn)$ time described in the next section.

\section{Optimal resolution of SSES}

In this section we prove the following theorem, which is a generalization to two dimensions of Theorem~16 in~\cite{GAWRYCHOWSKI2020}.

\begin{theorem}\label{theorem: main}
Given a submatrix equation system $E$ for an ${m\times n}$ matrix, we can find a generic solution $\phi(E)$ in optimal $O(|E|+mn)$ time.
\end{theorem}

In this section we consider a {\em typed} representations of $SSES$s. 
To represent any equation in $E$ we permit tuples to store the coordinates of {\em any} pair of homologue corners of the two matching squares and we add an additional component in $\{0,1,2,3\}$ to denote the reference corner.  
For example the equation $A[i..i+\ell-1][j..j+\ell-1]=A[i'..i'+\ell-1][j'..j'+\ell-1]$ can be stored as  $(i,j,i',j',\ell,0)$, where 0 denotes the upper-left corner, or as $(i+\ell-1,j+\ell-1,i'+\ell-1,j'+\ell-1,\ell,2)$ where 2 denotes the bottom-right corner.  
In the following we say that an equation $e=(i,j,i',j',\ell,t)$ with $t\in \{0,1,2,3\}$ has type $t$. The reason for introducing types, will be apparent after Definition~\ref{def:d1}.


\begin{definition}[Equation graph]
Given an $SSES$ $E$, the (typed) equation graph $G_{eq}(E)=(V,A)$ is the undirected labeled graph defined as follow. For every typed equation $(i,j,i',j',\ell,t) \in E$ we add the pairs $(i,j)$ and $(i',j')$ to the vertex set $V$ and the arc $(i,j)\leftrightarrow (i',j')$ with label ${\ell,t}$ to the arcs set $A$.
\end{definition}

Note that, contrary to $G_{pos}(E)$, the equation graph $G_{eq}(E)$ has size $\Theta(\vert E\vert)$.
Since there is a bijection between (typed) $SSES$s and (typed) equation graphs, we use these objects interchangeably and when clear from the context, we write $G_{pos}$ and $G_{eq}$ for the corresponding graphs induced by $E$. 



%


The next two definitions extend the notion of $k$-special and $k$-short equations in~\cite{GAWRYCHOWSKI2020} to two-dimensional square equations using base $b=16$.

\begin{definition}\label{def:kspecial}
An integer $i\geq 0$ is $k$-special in base $16$ if all the $k$ less significant hexadecimal digits of $i$ are non-zero. 
We denote with $\mathbb{S}_k$ the set of all the integers that are $k$-special in base $16$ and we define $\mathbb{S}_k^n=\mathbb{S}_k\cap [n]$ where $\mathbb{S}_0=\mathbb{N}$.
Note that $|\mathbb{S}_k^n|\leq (15/16)^k n$ (see~\cite[Fact 11a]{GAWRYCHOWSKI2020}) and $\mathbb{S}_{k+1}\subset \mathbb{S}_k$ hold.
\end{definition}






\begin{definition}\label{def:d1}
We say that an equation $e=(i,j,i',j',\ell,t)$ is $k$-short if $\ell\leq 2\cdot16^k$ and that $e$ is $k$-special if  $\{i,j,i',j'\}\subset \mathbb{S}_k$.
We say that a system of equations $E$ is $k$-special/$k$-short if all its equations are $k$-special/$k$-short. 
\end{definition}

\noindent

\noindent
A $k$-special equation $(i,j,i',j',\ell,t)$ has a constraint on the coordinates $i,j,i',j'$ of the reference corner $t$. 
Indeed, we use types because we need equations which are $k$-special with respect to any corner (see the proofs of Lemmas~\ref{lemma:specialsplit}~\ref{lemma:simplesplit}).

We say that $E$ is equivalent to $E'$, and write $E\equiv E'$, if $\Phi(E)=\Phi(E')$. 
We call $E$ type-acyclic if for every equation type $t$ the equation graph does not contain any cycle formed by type-$t$ edges.

\begin{lemma}\label{fact:type-a}
If $E$ is $k$-special and type-acyclic then $|E|< 4(15/16)^{2k} mn$. 
\end{lemma}
\begin{proof}
Since $E$ is is type-acyclic, for every $t\in \{0,1,2,3\}$, the number of type-$t$ edges in $G_{eq}$ is at most $\vert V\vert-1$ so $G_{eq}$ has at most $4\vert V\vert$ edges. Since $E$ is also $k$-special it is $\vert V \vert \leq \vert\mathbb{S}_k^m\vert \vert\mathbb{S}_k^n\vert \leq (15/16)^{2k} mn$ by Definition~\ref{def:kspecial}.\qed
\end{proof}

Before entering the technical details we sketch the idea behind the proof of Theorem~\ref{theorem: main} that follows the approach in~\cite{GAWRYCHOWSKI2020}.
Given $E$ we compute an equivalent system $E'$ with $O(mn)$ equations of size at most $2\times 2$. For such an $E'$ computing $\Phi(E')=\Phi(E)$ using its position graph, as described in Section~\ref{sec:notation}, takes optimal $O(mn)$ time. To build $E'$ we find smaller equivalent equations partitioning squares into smaller squares. 
Since this increases the number of equations the algorithm alternates two type of operations. 
The {\sc Split} operations lower the side length of the equations while {\sc Reduce} operations discard redundant equations and keeps their number under control. 
Redundant equations are discarded using maximum spanning forests (MSFs) since the shortest equation on a typed-cycle, is implied by the others (Lemma~\ref{lemma:MSF}).
To bound the overall cost, we restrict where equations are anchored: the splits procedures force the equations to have a corner located at a $k$-special position and we exploit the fact that $k$-special integers decays geometrically.
The core of the procedure is the recursive routine {\sc Shorten}$(E,k)$ which turns a $k$-special, type-acyclic system into an equivalent one which is also $k$-short.
For $k=0$ {\sc Shorten} accepts any equation system (all integers are 0-special) and returns equations of size at most~2 from which we can derive $\Phi(E)$ in $O(mn)$ time with the naive algorithm of Section~\ref{sec:notation}.


\begin{lemma}[{\sc Reduce}]\label{lemma:reduce}
Given an equation system $E$ we can compute in $O(|E|)$ time a type-acyclic system $E'\subseteq E $ such that $E' \equiv E$.
If $E$ is $k$-special then $|E'| \leq 4(15/16)^{2k} mn$.
\end{lemma}
\begin{proof}
We use Lemma~\ref{lemma:MSF} and set $E'=\bigcup_{t=0}^3 MSF(G^t_{eq})$. 
The time bound follows by observing that the MSF algorithm in~\cite{FREDMAN} takes linear time. To prove the second part, observe that since $E' \subseteq E$, if $E$ is $k$-special the same is true for $E'$ hence, by Lemma~\ref{fact:type-a}, it is $|E'| < 4(15/16)^{2k} mn$ as claimed.\qed
\end{proof}

The following pseudocode of {\sc Shorten} is that in~\cite{GAWRYCHOWSKI2020}, with our 2D-adapted implementations of the {\sc Split} (Lemmas~\ref{lemma:specialsplit}~\ref{lemma:simplesplit}) and {\sc Reduce} (Lemma~\ref{lemma:reduce}) operations.

\begin{algorithm}
\caption{\textsc{Shorten}$(E,k)$}
\label{alg:shorten}
\begin{algorithmic}[1]
\Function{Shorten}{$E,k$}  \Comment{$E$ must be $k$-special and type-acyclic}
    \If{$E = \emptyset$} \Return{$\emptyset$}
    \EndIf
    \State $(E_1,E_2) \gets \Call{SpecialSplit}{E,k}$ \label{line:specialsplit}
    \State $E_2' \gets \Call{Reduce}{E_2}$ \label{line:reduce1}
    \State $F \gets \Call{Shorten}{E_2',k+1}$ \label{line:reccall}
    \State $F' \gets \Call{SimpleSplit}{E_1 \cup F,k}$ \label{line:simplesplit}
    \State $F'' \gets \Call{Reduce}{F'}$ \label{line:reduce2}
    \State \Return{$F''$} \Comment{$F''$ is $k$-short, $k$-special and type-acyclic, and $F'' \equiv E$}
\EndFunction
\end{algorithmic}
\end{algorithm}

\begin{lemma}\label{lemma:shorten}
Algorithm~\ref{alg:shorten} is correct and terminates in $O(mn)$ time.
\end{lemma}
\begin{proof}
We prove the Lemma by backward induction on $k$.
The base case is any $k\geq\lfloor \log_{16}\min(m,n) \rfloor+2$, for which the algorithm is correct since $E=\emptyset$. 
Now consider a generic $k\geq 0$ and assume $E$ is type-acyclic and $k$-special. At line~\ref{line:specialsplit} {\sc SpecialSplit} (Lemma~\ref{lemma:specialsplit}) 
partitions every $k$-special equation into $O(1)$ equations that are either $(k+1)$-short or anchored at a ($k+1$)-special position.
The latter are pruned (Line~\ref{line:reduce1}) using Lemma~\ref{lemma:reduce} and passed to the recursive call at depth $k+1$, which makes them $(k+1)$-short (Line~\ref{line:reccall}). All the $k+1$-short pieces are collected and brought down to size $\leq 2\cdot16^k$ by the {\sc SimpleSplit} call (Lemma~\ref{lemma:simplesplit}), and a final {\sc Reduce} call restores type-acyclicity.
Since the number of equations alive at depth $k$ is bounded by $O((15/16)^{2k} mn)$, the work per level shrinks geometrically and for $k=0$ the whole recursion costs $O(mn)$ time.\qed

\end{proof}



\noindent{\bf Proof of Theorem~\ref{theorem: main}}
Given $E$, after an $O(mn)$ one-time pre-processing step that removes the need for any subsequent vertex relabeling as in~\cite{GAWRYCHOWSKI2020},
we compute $E'$ such that $E'\equiv E$ and $E'$ is type acyclic (and 0-special) in $O(\vert E\vert)$ time (Lemma~\ref{lemma:reduce}). 
Then, in $O(mn)$ time (Lemma~\ref{lemma:shorten}) we compute $E'' = \mbox{\sc Shorten}(E',0)$. By Lemma~\ref{lemma:shorten} $E''\equiv E'$, and $E''$ is 0-short (equations have size at most $2\times 2$) and type acyclic ($\vert E''\vert = O(mn)$). Hence, we can compute $\Phi(E'') = \Phi(E)$ in additional $O(mn)$ time by building $G_{pos}(E'')$ and computing its connect components as outlined in Section~\ref{sec:notation}.\qed

\section{Applications to repetitiveness measures}

\newcommand{\Eq}{\mathsf{Eq}}
\newcommand{\Ch}{\mathsf{Ch}}

As in~\cite{Shibata_Bannai_2026}, we use SSES to define a repetitiveness measure by introducing character assignments, and we compare this measure with 2D bidirectional macro schemes.  

\begin{definition}
    A SSES representation for a matrix $A\in\Sigma^{m\times n}$ is a quadruplet $(m,n,\Eq,\Ch)$ where $\Eq$ is a set of square submatrix equations, and $\Ch$ is a set of character assignments of the form $A[i][j]=c$ with $c\in \Sigma$, $0 \leq i < m$, $0 \leq j < n$.
    The size of the representation is $|\Eq| + |\Ch|$, and the measure $s(A)$ is defined as the size of the smallest SSES representation whose unique solution is the matrix~$A$. 
\end{definition}

For example, for the matrix in Fig.~\ref{fig:gpos}, it is easy to see that the smallest SSES representation is the one shown in the figure and has size 10.

\begin{lemma}\label{lemma:reduction}
For any valid squared bidirectional macro scheme $B$ for a matrix $A\in \Sigma^{m\times n}$, there exists an SSES representation $E$ of size $|B|$ whose unique solution is the matrix~$A$. Given $B$, we can compute $E$ in $O(\vert B\vert)$ time.
\end{lemma} 
\begin{proof}
For a copied $\ell\times\ell$ phrase $p\in B$ with origin and destination with top left corner in $(i,j)$ and $(i',j')$ respectively we add to $E$ the equation $(i,j,i',j',\ell)$. 
For any explicit symbol in $B$ we add to $E$ the corresponding character assignment. 
Since $B$ uniquely represents~$A$ the same is true for~$E$.\qed
\end{proof}


\begin{lemma}
    For every matrix $A$ it is $s(A) \leq \Bsq(A) \leq 2 s(A)$.
\end{lemma}
\begin{proof}
The inequality on the left derives from Lemma~\ref{lemma:reduction}. 
To prove the other one, given any SSES $E$ representation for $A$, we build a valid macro scheme $B$ of size $2\vert E\vert$ for $A$.
For every symbol assignment in $E$ we add the corresponding explicit symbol to $B$.
For every $(i,j,i',j',\ell)\in E$ we add to $B$ the two phrases $A[i..i+\ell-1][j..j+\ell-1]$ and $A[i'..i'+\ell-1][j'..j'+\ell-1]$, each using the other as its source.
$B$ is valid for $A$ since for every $(i,j)$ in $G_{pos}$ the path leading to a symbol assignment corresponds to a sequence $\{ (i_h, j_h)\}$ that satisfies~\eqref{eq:valid}.\qed
\end{proof}


Because of phrase overlapping, to decompress a 2D macro scheme $B$, a naive algorithm will visit the graph induced by the function $f_B\colon [m] \times [n]  \to 2^{[m] \times [n]} \cup \{\bot\}$. The visit could take time proportional to the size of the graph that in some cases could be $\Theta(\vert B\vert mn)$.
From  Lemma~\ref{lemma:reduction} and Theorem~\ref{theorem: main}, we get an optimal decompression algorithm.

\begin{corollary}
We can decompress any valid two-dimensional macro scheme of size $b$ of an $m\times n$ matrix in $O(b+mn)$ time.
\end{corollary}

\bibliographystyle{plain}
\bibliography{references}

\appendix

\section{Technical lemmas}

The next lemma extends the result in~\cite[Lemma 5]{GAWRYCHOWSKI2020} to our typed two-dimensional equation systems.

\begin{lemma}\label{lemma:MSF}
Let $E$ be an equation system and for $t=0,1,2,3$ let $G^t_{eq}$ be the subgraph of $G_{eq}$ formed by all its type-$t$ edges. 
It holds that $\bigcup_{t=0}^3 MSF(G^t_{eq})\equiv E$ where $MSF$ is the maximum spanning forest of a graph.
\end{lemma}
\begin{proof}
It suffices to prove that $MSF(G^t_{eq})\equiv E^t=\{e\in E\ \vert e\text{ has type } t\}$.
Consider any $e\in E^t$ not in $MSF(G^t_{eq})$. 
Adding $e$ to $MSF(G^t_{eq})$ must form a cycle $c_1 \leftrightarrow c_2 \ldots c_l\leftrightarrow c_1$ where every $c_i \leftrightarrow c_{j}$ has type $t$ and weight $l_i$ and $c_l \leftrightarrow c_1$ corresponds to $e$.
By the properties of the maximum spanning forest, it holds that $\ell_l\leq \min\{\ell_1,\ldots,\ell_{l-1}\}$, thus by transitivity $e$ is implied by equations in the path $P=c_1 \leftrightarrow c_2 \ldots c_{l-1}\leftrightarrow c_l$ of $MSF(G^t_{eq})$ since its squares are included in those of $P$ and have the same type.\qed
\end{proof}

\newcommand{\cifre}[2]{(#1\cdots #1)^{#2}_{16}}

In the following for $x\in \mathbb{N}$ we denote with $(x)_{16}$ the hexadecimal representation of $x$, and we denote with $(x)_{16}[i]$ the $i$-th digit, so that $x=\sum_{i=0}^\infty (x)_{16}[i]16^i$. 
For $j\in[0,15]$ and $k>0$ we write $\cifre{j}{k}$ to denote the positive integer consisting of $k$ hexadecimal digits equal to $j$, that is, $\cifre{j}{k} = j +16j+ \cdots + 16^{k-1}j$.

Finally, with $X+d$ where $X=x_1,\ldots,x_{\vert X \vert}$ and $d\in \mathbb{N}$ we denote the numbers in the set $\{x+d\ \vert\ x\in X\}$.

\begin{lemma}\label{lemma:l1}
Given non-negative integers $i_1,i_2,i_3,i_4$ and $y\in [0,11]$, let $z=y+4$. For any of the following statements and $k>0$ we can find in $O(k)$ time and $O(1)$ additional space an integer $d$ with $\cifre{y}{k} \leq d \leq \cifre{z}{k}$ such that the statement holds:
\begin{enumerate}[a)]
    \item\label{lemma:l1:1} $\{i_1,i_2,i_3,i_4\}+d$ are $k$-special;
    \item\label{lemma:l1:2} if $i_3,i_4 > \cifre{z}{k} $, then $\{i_1,i_2\}+d$ and $\{i_3,i_4\}-d$ are $k$-special;
    \item\label{lemma:l1:3} if $i_1,i_2,i_3,i_4> \cifre{z}{k}$ then $\{i_1,i_2,i_3,i_4\}-d$ are $k$-special.
\end{enumerate}
\end{lemma}

\begin{proof}
Recall that an integer is $k$-special if none of its less significant $k$ hexadecimal digits is zero. 
We exhibit the proof of~\ref{lemma:l1:2} since~\ref{lemma:l1:1} and~\ref{lemma:l1:3} are analogous.
We preliminary observe that given any $4$ hexadecimal digits $h_1,\ldots,h_4$ there exists a value $w$ in every range $[y,y+4]\subset[16]$, such that $h_{z_1} + w \bmod 16 \neq 0$ and $16+h_{z_2}-w\bmod 16\neq 0$ hold for every $z_1\in \{1,2\}$ and $z_2\in \{3,4\}$. 
To see this, note that there are 4 constraints and 5 possible values in $[y,y+4]$. 
We precompute for every possible value of $y$ an array storing the minimum admissible $w$ for every quadruple of hexadecimal digits $h_1,\ldots,h_4$.

Next, we derive the desired $d$ by computing in $O(k)$ time its hexadecimal digits right to left starting from the less significant. The bounds on the size of $d$ will follow from the observation that it consists of $k$ hexadecimal digits between $y$ and $z = y+4$.  
Let $d_i$ denote the value of $d$ obtained after we computed $i$ digits.


We prove the correctness of our procedure by induction on $k$.
For $k=1$, with the precomputed arrays we find in $O(1)$ time $w$ such that $\{i_1,i_2\}+w$ and $\{i_3,i_4\}-w$ are $1$-special by using the four digits $h_1 = (i_1)_{16}[0],\ldots, h_4 = (i_4)_{16}[0]$ and we set $d_1=w$.
Note that $w\in [y, z]$ so borrowing is possible, if needed, since $i_3,i_4> z$ by hypothesis.
Now assume the statement holds for $k>1$, in particular $d_k \leq \cifre{z}{k}$ holds by induction.
Similarly to the base case, we retrieve $w$ from the precomputed arrays, using the $(k+1)$-st digit of $(\{i_1,i_2\}+d_k)_{16}$ and $(\{i_3,i_4\}-d_k)_{16}$.
By hypothesis it is $i_3,i_4> \cifre{z}{k+1}$ thus borrowing, if needed, is possible. We set $d_{k+1}=d_k+w16^k$.
The values $\{i_1,i_2\}+d_{k+1}$ and $\{i_3,i_4\}-d_{k+1}$ are $(k+1)$-special since by adding/subtracting $w16^k$ to the values obtained in the previous iteration, we make all the $(k+1)$-st digits non-zero without modifying any digit at positions $0,\ldots,k-1$ which are all non-zero by the induction.\qed
\end{proof}

Next we show that~\cite[Lemma 12]{GAWRYCHOWSKI2020} holds also in two dimensions with our definition of 2D $k$-special and $k$-short equations (see Definition~\ref{def:d1}), but with a multiplicative $O(k+1)$ time factor.

\begin{lemma}[{\sc SpecialSplit}]\label{lemma:specialsplit}
Given any set $E$ of $k$-special equations, we can build in $O((k+1)\vert E\vert)$ time two sets of equations $E_1,E_2$ such that: 1) $E_1$ is $(k+1)$-short and $k$-special 2) $E_2$ is $(k+1)$-special 3) $E_1 \cup E_2 \equiv E$ and 4) $\vert E_1\vert +\vert E_2\vert =O(\vert E\vert)$.
\end{lemma}

\begin{proof}
To prove the claim we split every $e\in E$ in $O(k)$ time into $O(1)$ equations that together are equivalent to $e$. 
Each such equation satisfies condition $1)$ or $2)$ and is added to $E_1$ or $E_2$ accordingly.
Let $e=(i,j,i',j',\ell,t)$ be a $k$-special equation in $E$ and assume $t=0$, the proofs for $t=1, 2, 3$ are analogous.
{If $e$ is already $(k+1)$-short, we add it to $E_1$. 
If this is not the case, but one of its corners is $(k+1)$-special, 
we change the type of $e$ accordingly (if needed) and we add it to $E_2$. 
In either case, we are done.}
Otherwise we replace $e$ with $4$ new equations $e_1,e_2,e_3,e_4$ that together are equivalent to $e$.
Equation $e_1$ is $k$-special and $(k+1)$-short so we will add it to $E_1$ while the remaining ones are all $(k+1)$-special and thus we will add them to~$E_2$. 
Let $d$ denote the integer obtained in $O(k)$ time by applying Lemma~\ref{lemma:l1} part~\ref{lemma:l1:1} to the integers 
$i,j,i',j'$ with $y=5$ and parameter $k+1$. Note that by the lemma it is $d \leq \cifre{9}{k+1} <  10 \cdot 16^k$; since $e$ is not $(k+1)$-short, $\ell>2\cdot16^{k+1}>2d$.
We set $e_1=(i,j,i',j',d,0)$, thus $e_1$ is $k$-special and $(k+1)$-short as claimed. Next, we set $e_2=(i+d,j+d,i'+d,j'+d,\ell-d,0)$, which by Lemma~\ref{lemma:l1} is $(k+1)$-special.

The derivation of the last two equations is more complex: they are of type 3 and 1 and we need them to cover respectively the rectangles $R_{BL}$ and $R_{TR}$ of Fig.~\ref{fig:alpha}
to ensure that $\{e_1,\ldots,e_4\} \equiv e$.
To compute $e_3$ we consider the four integers 
$
i,j+\ell-1, i', j'+\ell-1
$
which are the coordinates of the upper right corners of~$e$. Applying Lemma~\ref{lemma:l1} part~\ref{lemma:l1:2} with $y=0$ and parameter $k+1$ we find an integer $\alpha$ such that
$$
i + \alpha,j+\ell-1-\alpha, i' + \alpha, j'+\ell-1 - \alpha
$$
are $(k+1)$-special. Hence if we set 
$$
e_3 = (i + \alpha,j+\ell-1-\alpha, i' + \alpha, j'+\ell-1 - \alpha,\ell-\alpha,3)
$$
by construction $e_3$ is $(k+1)$-special. By Lemma~\ref{lemma:l1} we also have 
$
\alpha \leq \cifre{4}{k+1} < \cifre{5}{k+1} \leq d
$
which ensures that $e_3$ covers $R_{BL}$, i.e. that $\ell-\alpha \geq \max(\ell-d,d)$.  The construction of $e_4$ to cover $R_{TR}$ is analogous starting with the lower right corners of $e$.\qed
\end{proof}

\begin{figure}[t]
\centering
\begin{subfigure}[t]{0.48\textwidth}
    \centering
    \includegraphics[height=0.28\textheight,keepaspectratio]{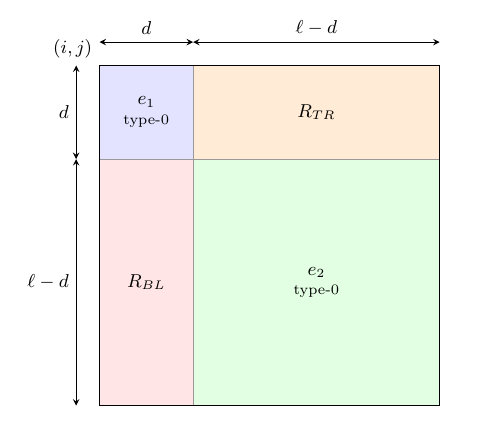}
    \label{fig:special_split1}
\end{subfigure}\hfill
\begin{subfigure}[t]{0.48\textwidth}
    \centering
    \includegraphics[height=0.28\textheight,keepaspectratio]{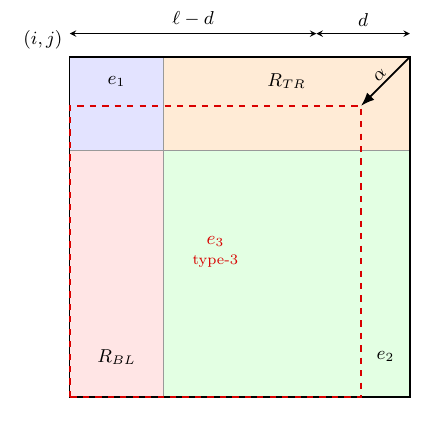}
    \label{fig:special_split2}
\end{subfigure}
\caption{Construction of the equations $e_1, e_2, e_3$ used in Lemma~\ref{lemma:specialsplit} that together with $e_4$ (not shown) are equivalent to~$e=(i,j,i',j',\ell,0)$. The requirement is that the union of the submatrices associated with $e_1,\ldots,e_4$ coincides with the submatrix associated to $e$. The parameters $d$ and $\alpha$ are chosen to ensure that $e_1$ is $k$-short and the reference corner $v$ of the type-$v$ equations $e_2,e_3$ are $(k+1)$-special.\label{fig:alpha}} 
\end{figure}

The next lemma extends~\cite[Lemma 13]{GAWRYCHOWSKI2020} to two-dimensional equation systems.
Again we have an additional $O(k)$ multiplicative factor in the time.

\begin{lemma}[{\sc SimpleSplit}]\label{lemma:simplesplit}
Given any set $E$ of $k$-special and $(k+1)$-short equations, we can build in $O(k\vert E\vert)$ time an equivalent $k$-special and $k$-short system of size $O(\vert E\vert)$. 
\end{lemma}
\begin{proof} 
Let $e=(i,j,i',j',\ell,t)$ be any equation in $E$ that we assume w.l.o.g. to be of type $t=0$. 
The proofs for the other types are symmetrical.
If $e$ is already $k$-short we have done, otherwise $2\cdot 16^k< \ell \leq 2\cdot16^{k+1}$ hold and we cover the two matching squares $S_e$ and $S_e'$ corresponding to $e$ as follows.
We add $O(1)$ $k$-short equations of the form $(i+z_116^k,j+z_216^k,i'+z_116^k,j'+z_216^k,16^k,0)$ for every integers $z_1,z_2\in[0,\lfloor\ell/16^k\rfloor-1]$. 
These equations are $k$-special because $\{i,j,i',j'\}\subset \mathbb{S}_k$ and adding multiples of $16^k$ does not change the $k$-specialty of the coordinates.
If $\ell \bmod 16^k=0$, the above equations partition $S_e$ and $S_e'$ and are equivalent to $e$.
Thus next we assume $k>0$ and $\ell \bmod 16^k\neq0$.
In this case, $O(1)$ rectangles are left uncovered in the bottom and right edges of $S_e$ and $S_e'$, plus $2$ squares in their bottom-right corners.
To cover the squares we add a type-$0$ equation matching them which is already $k$-special and $k$-short.
Now consider any two matching rectangles $R$ and $R'$ in the bottom-edge of $S_e$ and $S_e'$, excluding the rightmost ones.
These rectangles have size $l\times 16^k$ with $l=\ell \bmod 16^k$ and their bottom-left corners are anchored at the positions $(i+\ell-1,j+z_216^k)$ and $(i'+\ell-1,j'+z_216^k)$ respectively, for some $z_1,z_2$.
To cover $R$ and $R'$, we start from the $16^k\times 16^k$ squares $Q,Q'$ with top-right corners at $(v,h)=(i+\ell-16^k,j+(z_2+1)16^k-1)$ and $(v',h')=(i'+\ell-16^k,j'+(z_2+1)16^k-1)$, that respectively include $R$ and $R'$. 
Then by using Lemma~\ref{lemma:l1} part~\ref{lemma:l1:2} with $y=0$ and parameter $k$, we find a nonnegative $\alpha<5\cdot 16^{k-1}$ such that the integers $v-\alpha$, $h+\alpha$, $v'-\alpha$ and $h'+\alpha$ are all $k$-special.
The resulting $k$-special and $k$-short equation is $(v-\alpha, h+\alpha, v'-\alpha, h'+\alpha,16^k+\alpha,3)$.
Similarly, to cover the $2$ rightmost rectangles in the bottom-edge of $S_e$ and $S'_e$, we still consider the squares $Q$ and $Q'$ including them, but now with $(v,h)$ and $(v',h')$ we refer the coordinates of their top-left corners.
We find an $\alpha<5\cdot 16^{k-1}$ using Lemma~\ref{lemma:l1} part~\ref{lemma:l1:3} with $y=0$ and parameter $k$ such that $v-\alpha$, $h-\alpha$, $v'-\alpha$ and $h'-\alpha$ are all $k$-special.
In this case the added equation is $(v-\alpha, h-\alpha, v'-\alpha,h'-\alpha,16^k+\alpha,0)$ which is $k$-special and $k$-short.
See Figure~\ref{fig:simple_split} for a graphical representation of the partitioning.
We can handle similarly all the $O(1)$ rectangles on the right edge of $S_e$ and $S'_e$.
With the above procedure we add $O(\vert E\vert)$ equations in total, each generated in $O(k)$ time, thus the claimed time and space follows.

\end{proof}

\begin{figure}[!tp]
\centering
\begin{subfigure}[t]{0.48\textwidth}
    \centering
    \includegraphics[height=0.28\textheight,keepaspectratio]{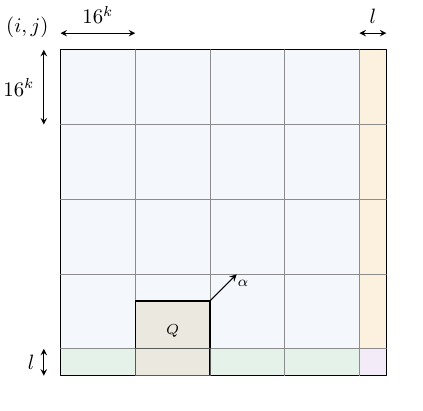}
    \label{fig:simple_split1}
\end{subfigure}\hfill
\begin{subfigure}[t]{0.48\textwidth}
    \centering
    \includegraphics[height=0.28\textheight,keepaspectratio]{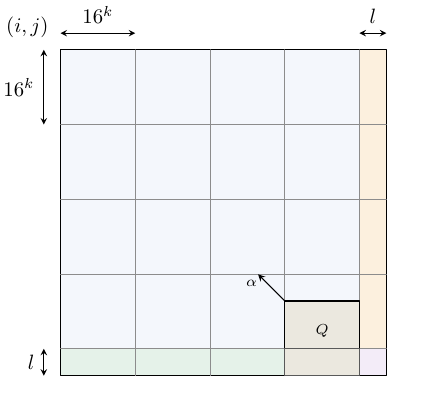}
    \label{fig:simple_split2}
\end{subfigure}
\caption{
Partitioning of a $k$-special and $k+1$-short equation $e$ as described in Lemma~\ref{lemma:simplesplit}, here $l=\ell \bmod 16^k$ where $\ell$ is the equation length. The two squares in the figure represents $S_e$ i.e. one of the two matching squares corresponding to $e$.
The $l\times l$ square in the bottom-right corner and all the $16^k\times 16^k$ squares are $k$-special and $k$-short.
In the figure on the right, the square $Q$, which covers the rightmost rectangle in the bottom edge, is extended along its anti-diagonal while the remaining rectangles are covered as shown in the figure on the left. 
\label{fig:simple_split}} 
\end{figure}

\end{document}